\documentclass{article} 
\usepackage[font=footnotesize, labelfont=bf]{caption}
\usepackage{amsmath,amsfonts,amssymb,amsthm}
\usepackage{hyperref}
\usepackage{multirow}
\newtheorem{remark}{Remark}
\newtheorem{proposition}{Proposition}

\usepackage[sort&compress,nameinlink,noabbrev,capitalize]{cleveref}
\newif\ifarxiv
\arxivfalse
\usepackage{etoolbox}
\usepackage{nicefrac}
\usepackage{booktabs, adjustbox}
\newcommand{\appsymb}{$\bigstar$}
\newcommand{\appref}[1]{{\appsymb}}

\usepackage{xcolor}

\newcommand{\gas}{\text{gas}}
\newcommand{\stake}{\text{stake}}
\newcommand{\gasprice}{\text{gasprice}}

\begin{document}
\pagestyle{plain}
\title{An Adaptive Multichain Blockchain: \\ A Multiobjective Optimization Approach}
\author{%
\begin{tabular}[t]{c}
Nimrod Talmon
\\[2pt]
{\small Ben-Gurion University}\\[1pt]
{\small \texttt{talmonn@bgu.ac.il}}
\end{tabular}
\and
\begin{tabular}[t]{c}
 Haim Zysberg\\[2pt]
{\small Independent}\\[1pt]
{\small \texttt{haimzys@gmail.com}}
\end{tabular}
}
\maketitle\maketitle

\begin{abstract}
Blockchains are widely used for secure transaction processing, but their scalability remains limited, and existing multichain designs are typically static even as demand and capacity shift. We cast blockchain configuration as a multiagent resource-allocation problem: applications and operators declare demand, capacity, and price bounds; an optimizer groups them into \emph{ephemeral chains} each epoch and sets a chain-level clearing price. The objective maximizes a governance-weighted combination of normalized utilities for applications, operators, and the system. The model is modular—accommodating capability compatibility, application-type diversity, and epoch-to-epoch stability—and can be solved off-chain with outcomes verifiable on-chain. We analyze fairness and incentive issues and present simulations that highlight trade-offs among throughput, decentralization, operator yield, and service stability.
\end{abstract}

\section{Introduction}

Blockchains provide transparency and security but remain rigid in design and operation, which limits their ability to respond to dynamic workloads, heterogeneous applications, and operator volatility. Multichain architectures, such as Polkadot~\cite{wood2016polkadot} and COSMOS~\cite{kwon2019cosmos} address scalability by distributing load across parallel chains, yet they are typically static: chains are predefined and remain fixed regardless of changing demand or capacity (and apps and operators join them manually). As a result, performance may degrade, scalability is constrained, gas pricing may become imbalanced~\cite{roughgarden2021transaction}, and user experience is restricted, as applications cannot rely on the architecture to adapt transparently to their needs (causing also a likely underutilization of operator resources).  
This rigidity arises both because chains are interdependent (via shared security, messaging, and tooling) and from governance inertia, making real-time adaptation difficult even when spare capacity exists somewhere in the system.

In practice, this rigidity is increasingly visible: when activity spikes on a few high-demand applications, transaction fees rise network-wide while other chains remain idle; operators must often reconfigure infrastructure manually, leading to delays and inconsistent pricing; and, moreover, long-lived chain assignments prevent efficient matching between applications and operators, hindering innovation and reducing resilience to failures or governance changes.
These frictions highlight the need for architectures that can self-adjust and reallocate resources dynamically. (E.g., dynamic sharding~\cite{tennakoon2022dynamic}.)

We propose a framework for adaptive multichain blockchains that reconfigure in response to changing demand and operator capacity, while also incorporating forecasts and governance decisions. 
At its core is a multiobjective optimization model in which applications, operators, and users are treated as agents with distinct objectives and engage in bids reflecting resource requirements and availability. 
The optimization mechanism groups the applications and operators into ephemeral chains with shared parameters such as gas prices and operational constraints, and integrates application, operator, and system utilities to balance heterogeneous interests.  
These reconfigurations occur periodically or in response to performance triggers, enabling the system to adapt to real-time fluctuations while preserving stability across epochs.
%
By combining economic bidding, algorithmic optimization, and policy-level control, the framework provides a principled foundation for self-organizing blockchain infrastructures.

We formalize this multiagent model, analyze its properties with respect to adaptivity, efficiency, and strategyproofness, and demonstrate its modular design through extensions capturing service degradation, stability, and application diversity. We analyze certain normative aspects, including fairness and strategyproofness, and show feasibility through simulations illustrating tradeoffs between agent utilities. 
Our experiments demonstrate computational feasibility at realistic scales and reveal how governance parameters can modulate the relative influence of applications, operators, and the system, providing direct control over multiagent priorities and tradeoffs.

Our results highlight how adaptive optimization can reconcile competing objectives under realistic constraints, revealing emergent equilibria between stakeholder groups.
Beyond specific design choices, the framework offers a general methodology for reasoning about coordination in dynamic, partially decentralized environments.
By treating blockchain configurability as a multiagent resource allocation problem, the work contributes to the study of coordination and mechanism design in multiagent systems.
To our knowledge, this is the first formal model that treats chain formation itself as a multiobjective optimization balancing application, operator, and system utilities, providing a general methodology for adaptive blockchain configuration.

\paragraph{Scope and contribution}
This paper focuses on a single-epoch optimization problem underlying our adaptive multichain configuration. The goal is to: (1) introduce a multiagent architecture in which applications, operators, and the system are modeled as interacting stakeholders; (2) formalize a modular optimization model for forming ephemeral chains under feasibility constraints; and (3) demonstrate baseline computational feasibility and governance-adjustable trade-offs in simulation. Broader directions—such as full multi-epoch equilibrium analysis, alternative normalization and aggregation schemes, richer solver–market couplings, and stronger incentive guarantees—are important but left for future work. See Section~\ref{section:outlook} for a more thorough discussion on future research directions.

\section{Related Work}

Our work connects to research on multiagent resource allocation, coalition formation, and mechanism design~\cite{shehory1998task,sandholm1999coalition,brandt2016handbook}. In our setting, applications and operators act as agents with heterogeneous objectives, and assignments determine both feasibility and induced gas prices. This relates to strategyproof mechanism design and auction-based allocation~\cite{nisan2007algorithmic,parkes2001iterative,clarke1971multipart,groves1973incentives}, as well as fairness notions such as Dominant Resource Fairness (DRF)~\cite{ghodsi2011drf} and lottery-based allocation~\cite{kleinberg2018fairness,segalhalevi2020fair}. 
Vickrey–Clarke–Groves (VCG) mechanisms are a canonical benchmark for truthful reporting, but they are a poor fit here. Our allocation is a combinatorial \emph{chain-formation} problem with hard feasibility constraints, making welfare optimization (and pivotal payments) expensive at scale (NP-hard in particular, as we show below), and VCG additionally requires explicit transfers that must be funded and enforced within protocol economics. Moreover, the mechanism is inherently repeated across epochs, so one-shot truthfulness is not the sole design target. Accordingly, we do not pursue VCG-style guarantees; instead, we rely on a transparent optimization objective with governance-controlled weights, coupled with accountability (collateral-backed penalties and auditing) to discourage profitable misreporting over time.

Our work is tightly coupled with blockchain scalability via multichain designs, including relay-chain architectures such as Polkadot~\cite{wood2016polkadot}, interoperability frameworks such as Cosmos~\cite{kwon2019cosmos}, Ethereum rollups~\cite{buterin2020eth2}, and ICP subnets~\cite{williams2021icp}. 
These approaches improve throughput but rely on static or governance-driven provisioning, whereas our framework introduces ephemeral, optimization-driven chains that evolve automatically based on observed demand and operator capacity, integrating economic and governance feedback into the reconfiguration process.
To our knowledge, this is the first model that treats chain formation itself as a multiobjective optimization problem balancing application, operator, and system utilities.

Finally, our use of predictive modeling connects to research on self-adaptive systems~\cite{weyns2012survey,wong2022self} and ML-based scheduling and load prediction~\cite{chen2018machine}.
Machine learning in our setting enhances adaptivity rather than replacing optimization: forecasts of gas usage and operator performance are used to guide reallocation decisions, enabling responsiveness without centralized control.

\section{Preliminaries}

We outline the blockchain context for our model and its framing as a multiagent resource allocation problem.

\subsection{Blockchain Setting}

A blockchain is a distributed ledger that records transactions and executes applications commands. Here we concentrate on Proof of Stake blockchains, which are maintained by operators, who contribute computational capacity and lock stake as a security guarantee. 
Applications submit transactions that consume \emph{gas}, a standardized unit of computation, and attach fees to have them processed.
Gas prices are typically determined by congestion and market bidding, creating dynamic tension between demand, operator capacity, and user costs.
Different chains may emphasize different properties, such as general-purpose execution, privacy-preserving transactions, or support for zero-knowledge proofs.
In multichain ecosystems, these variations give rise to heterogeneous workloads and incentives, motivating the need for coordination and adaptive allocation across chains.

We distinguish three classes of stakeholders in a blockchain: 
(1) applications, which demand gas under budget, security, and operational constraints; (2) operators, which supply capacity and stake subject to their capabilities; and (3) the system at large, representing users and governance, which values overall performance, security, decentralization, and diversity of applications. This perspective motivates casting the problem as a resource-allocation framework, as developed in the next section.

\subsection{System Dynamics}

Most blockchain systems evolve in discrete \emph{epochs}, each serving as a checkpoint where configurations may be updated (due to agents monitoring of the runtime dynamics and changing preferences). 
At these points (following some computational latency), applications can be reassigned, operators can change roles, and resources can be redistributed in response to shifting demand and capacity. Such reconfiguration enables adaptation but also incurs costs, including temporary downtime and reduced stability.

From a multiagent perspective, epochs can be seen as repeated rounds of allocation, where agents may change their preferences and the optimizer must weigh the possible gains from improved matching against the disruptions caused by change.

\subsection{Practical Assumptions}

Our framework relies on a small set of practical assumptions:
\begin{itemize}
    \item \textbf{Distribution and consensus.}
    We assume the mechanism is distributed as a protocol and configurations disseminated as signed \emph{configuration proposals} for the next epoch. A proposal becomes authoritative once it is finalized on-chain at the epoch boundary (e.g., via a governance/checkpoint transaction).

    \item \textbf{Identifiability and accountability.}
    We assume applications and operators are identifiable at the relevant protocol layer,
    e.g., via a registry that binds an application/operator identifier to one or more public keys (and, when needed, to an on-chain deposit). Identifiability is used only to support declarations (gas demand/capacity, stake, and price bounds) and accountability.
    Furthermore, we assume declarations are backed by collateral and enforceable via protocol penalties. In particular, persistent infeasibility, non-delivery (e.g., an operator failing to provide declared capacity), or detectable manipulation of declared parameters can trigger slashing or other penalties drawn from posted collateral.     

    \item \textbf{Predictive and auditing layer.}
    We assume observable on-chain data can be used to complement declarations, e.g., via forecasts of gas usage and operator performance. The same data can support governance-supervised auditing,
    including anomaly detection for suspicious usage patterns or repeated feasibility misreports.    

    Concretely, ML can be used (1) to predict near-term demand and operator performance, (2) to impute missing or noisy declarations from historical behavior, (3) to propose governance-level adjustments to parameters (e.g., $\gamma$ for overprovisioning or the relative weights used in normalization/penalties) based on recent operating conditions, and (4) to flag anomalous patterns indicative of slack exploitation, coordinated under-reporting, or repeated feasibility misreports. (The design and evaluation of such ML machinery is outside the scope of the current paper.)

    Technically, we assume that such ML can be executed off-chain, with outcomes made verifiable on-chain (see discussion later as well).
    
    \item \textbf{Global optimization feasibility.} 
    We assume that a global optimization can be executed off-chain, with outcomes made verifiable on-chain (e.g., via zero-knowledge proofs; exact details are out of scope). Furthermore, appropriate incentives and penalties ensure 
    the integrity of this verification layer.
\end{itemize}

These assumptions reflect mechanisms already common in blockchain practice,\footnote{Except for application identifiability, which is not so common, but we provide for it.} such as staking declarations, on-chain transparency, and the use of off-chain computation with cryptographic verification.\footnote{The latter is still ongoing research, but we have a full epoch between optimization runs, which makes our situation easier.}

\section{Core Model of Ephemeral Chains}

We now formalize the core model, in which applications and operators are assigned to ephemeral chains at each epoch; later we describe modular extensions to it.

\subsection{Notation}

We have applications, operators, and ephemeral chains: 
$\mathcal{APP}$ denotes the set of applications, 
$\mathcal{OP}$ the set of operators, 
and $\mathcal{CHAIN}$ the set of ephemeral chains formed in a given epoch.

Each chain is implicitly defined by the applications and operators assigned to it; chains have no independent identity outside these assignments.

\subsection{Decision Variables}

For each application $a \in \mathcal{APP}$ and chain $c \in \mathcal{CHAIN}$, 
let $x_{a,c} \in \{0,1\}$ indicate whether $a$ is assigned to $c$. 
For each operator $o \in \mathcal{OP}$ and chain $c \in \mathcal{CHAIN}$, 
let $y_{o,c} \in \{0,1\}$ indicate whether $o$ is assigned to $c$ (note that the number of ephemeral chains is naturally upper bounded by $\min(|\mathcal{OP}|, |\mathcal{APP}|)$, and typically lower as we bundle apps and operators together to form a single ephemeral chain).
Each agent (apps/ops) may be assigned to at most one chain.\footnote{For apps this is intuitive; for operators, we simplify by considering each operator as a single computational entity (this assumption may be lifted).}

\subsection{Agent Inputs}

Each application $a \in \mathcal{APP}$ declares its:
\begin{itemize}
    \item gas demand, denoted $\text{gas}_{a} > 0$,
    \item required stake, denoted $\text{stake}_{a} > 0$,
    \item maximum acceptable gas price, denoted $\text{gasprice}_{a} > 0$.
\end{itemize}
Each operator $o \in \mathcal{OP}$ declares its:
\begin{itemize}
    \item gas capacity, denoted $\text{gas}_{o} > 0$,
    \item available stake, denoted $\text{stake}_{o} > 0$,
    \item minimum acceptable gas price, denoted $\text{gasprice}_{o} > 0$.
\end{itemize}

These declarations can be viewed as preferences or bids: 
applications bid for gas under budget constraints, 
while operators bid to supply capacity under compensation requirements.

\subsection{Derived Chain Properties}

Given some assignment, it implicitly declares the ephemeral chains that are opened; for each chain $c \in \mathcal{CHAIN}$ we derive:
\begin{itemize}
    \item Applications assigned: $$\mathcal{APP}_{c} = \{ a \mid x_{a,c}=1 \}\ ,$$
    \item Operators assigned: $$\mathcal{OP}_{c} = \{ o \mid y_{o,c}=1 \}\ ,$$
    \item Total stake: $$\text{Stake}_{c} = \sum_{o \in \mathcal{OP}_{c}} \text{stake}_{o}\ ,$$
    \item Total demand: $$\text{Demand}_{c} = \sum_{a \in \mathcal{APP}_{c}} \text{gas}_{a}\ ,$$
    \item Total supply:\footnote{We use $\min_{o\in\mathcal{OP}_c}\text{gas}_o$ as a simple, conservative proxy for protocol-limited throughput:
in fully synchronized block production, progress is bounded by the slowest participating operator, so summing or averaging
would overestimate throughput by assuming independent parallelism. More generally, the bottleneck depends on the consensus
protocol and deployment (e.g., slowest supermajority/quorum for finality, committee rules, or network-latency effects).
Any such variant can be captured by replacing $\min$ with a monotone ``effective supply'' aggregator over
$\{\text{gas}_o:o\in\mathcal{OP}_c\}$, without changing the crux of the optimization and incentive structure.}
 $$\text{Supply}_{c} = \min_{o \in \mathcal{OP}_{c}} \text{gas}_{o}\ ,$$
    \item Gas actually processed: $$\text{Gas}_{c} = \min(\text{Demand}_{c}, \text{Supply}_{c})\ .$$
\end{itemize}

These properties are induced by the assignments and declared inputs, 
and characterize the capacity of each ephemeral chain.

\subsection{Constraints}

Each assignment must satisfy basic feasibility conditions.

\paragraph{Assignment.} 
Each agent can join at most one chain: formally,
$\sum_{c \in \mathcal{CHAIN}} x_{a,c} \leq 1$ for each $a \in \mathcal{APP}$ and $\sum_{c \in \mathcal{CHAIN}} y_{o,c} \leq 1$ for each $o \in \mathcal{OP}$.

\paragraph{Stake feasibility.}
The total operator stake on each chain must cover the maximum requirement among applications on that chain: formally, $\text{Stake}_{c} \;\geq\; \max_{a \in \mathcal{APP}_{c}} \text{stake}_{a}$.
This reflects security requirements: applications request a minimum level of stake-supplied collectively by operators to protect their execution.\footnote{This is a simplification: applications declare the stake they would require if placed alone on a chain, yet the actual security demand naturally grows with the number and size of co-located apps. A more involved definition could scale requirements accordingly, but we leave this extension for future work.}

\paragraph{Gas price feasibility.}
We introduce a decision variable $\text{gasprice}_{c}$ for each chain $c$, that must lie between operator requirements and application offers (for nonempty chains): $\max_{o \in \mathcal{OP}_{c}} \text{gasprice}_{o} 
\;\leq\; \text{gasprice}_{c} 
\;\leq\; \min_{a \in \mathcal{APP}_{c}} \text{gasprice}_{a}$.
This variable plays the role of a \emph{clearing price}: applications bid by declaring their maximum acceptable gas price, 
operators set asks via their minimum acceptable price, and the chain price must fall within the feasible interval.

\subsection{Utilities}

We define utilities for applications, operators, and the system.

\paragraph{Application utility.}
Each application $a$ derives value from the fraction of its declared demand that is actually computed.\footnote{Indeed, applications also prefer low gas prices, and a penalty term can naturally be included in their utility. For methodological clarity, we introduce this only later, once the multiobjective structure is in place.} We model proportional allocation when chain demand exceeds supply:
\[
U^{\text{app}}_{a} \;=\; 
\sum_{c \in \mathcal{CHAIN}} x_{a,c}\;\frac{\text{gas}_{a}}{\text{Demand}_{c}}\;\text{Gas}_{c}.
\]
(Note that this equals $0$ if $a$ is unassigned.)

\paragraph{Operator utility.}
Each operator $o$ values its \emph{yield} (per stake), namely fees earned per unit of staked capital.
Let the total fee on chain $c$ be $\text{Fee}_{c}=\text{gasprice}_{c}\cdot \text{Gas}_{c}$ and recall
$\text{Stake}_{c}=\sum_{o\in\mathcal{OP}_{c}}\text{stake}_{o}$.
We model stake-proportional revenue sharing within a chain, i.e., operator $o$ receives the fraction
$\text{stake}_{o}/\text{Stake}_{c}$ of $\text{Fee}_c$.\footnote{Note that different revenue distributions can be implemented within the model; see discussion later.} Thus,
\[
U^{\text{op}}_{o} \;=\;
\sum_{c\in\mathcal{CHAIN}}
y_{o,c}\cdot
\frac{\text{Fee}_{c}\cdot (\text{stake}_{o}/\text{Stake}_{c})}{\text{stake}_{o}}
\;=\;
\sum_{c\in\mathcal{CHAIN}}
y_{o,c}\cdot
\frac{\text{Fee}_{c}}{\text{Stake}_{c}}.
\]
(Note that this equals $0$ if $o$ is unassigned.)

\paragraph{System utility.}
At the system level we measure aggregate throughput as total fees (as this corresponds to user usage):
\[
U^{\text{sys}} \;=\; \sum_{c \in \mathcal{CHAIN}} 
\text{gasprice}_{c} \cdot \text{Gas}_{c}.
\]

\subsection{Objective}

To make the different utilities comparable, we normalize them to the range $[0,1]$ 
using bounds; these bounds can be obtained analytically from the model 
(e.g., maximum application utility when demand is fully met, or maximum system utility when all gas is processed at the highest feasible price), or estimated empirically from historical runs to avoid unrealistic scaling.

Let $\widehat{U}^{\text{app}}_{a}$, $\widehat{U}^{\text{op}}_{o}$, 
and $\widehat{U}^{\text{sys}}$ denote the normalized utilities. 
The global objective is then to maximize a weighted combination:
\[
\max \;\;
\lambda_{\text{app}} \cdot \sum_{a \in \mathcal{APP}} \frac{\widehat{U}^{\text{app}}_{a}}{|\mathcal{APP}|}
\;+\;
\lambda_{\text{op}} \cdot \sum_{o \in \mathcal{OP}} \frac{\widehat{U}^{\text{op}}_{o}}{|\mathcal{OP}|}
\;+\;
\lambda_{\text{sys}} \cdot \widehat{U}^{\text{sys}},
\]
where $\lambda_{\text{app}}, \lambda_{\text{op}}, \lambda_{\text{sys}}$ are governance parameters, determining the relative importance of applications, operators, and system performance; we require  $\lambda_{\text{app}}, \lambda_{\text{op}}, \lambda_{\text{sys}} \geq 0$ 
and $\lambda_{\text{app}} + \lambda_{\text{op}} + \lambda_{\text{sys}} = 1$.

\subsection{Theoretical Intractability}

The resulting problem is a structured multiobjective optimization, where the different objectives capture the perspectives of applications, operators, and the system, and their relative weights can be tuned by governance to reflect policy priorities.

Not surprisingly, it is generally NP-hard, as we show next.

\begin{proposition}[NP-hardness]
Even in the core model, the following decision problem is NP-hard:
given an instance and a threshold $T$, decide whether there exists a feasible
assignment with $U^{\text{sys}}\ge T$ (for $\lambda_{\text{sys}}=1$).
\end{proposition}

\begin{proof}
We reduce from the NP-hard problem \textsc{Partition}~\cite{biblenp}. Given an instance of \textsc{Partition} - containing numbers $a_1,\dots,a_n$ with $\sum_{i=1}^n a_i = 2B$, and where the task is to decide whether the numbers can be partitioned into two groups, each summing up to $B$ - construct an instance with applications
$\mathcal{APP}=\{1,\dots,n\}$ and operators $\mathcal{OP}=\{1,2\}$.
For each application $i\in\mathcal{APP}$ set
$\text{gas}_{i}=a_i$, $\text{stake}_{i}=0$, and $\text{gasprice}_{i}=1$.
For each operator $j\in\mathcal{OP}$ set
$\text{gas}_{j}=B$, $\text{stake}_{j}=1$, and $\text{gasprice}_{j}=1$.
Set $\lambda_{\text{sys}}=1$ (and $\lambda_{\text{app}}=\lambda_{\text{op}}=0$), and set $T=2B$.

Consider any feasible assignment.
Because each operator may be assigned to at most one chain, at most two chains
can have $\mathcal{OP}_c\neq\emptyset$; all other chains (if any) have no operators
and contribute $0$ to $U^{\text{sys}}$.
Moreover, on any chain with $\mathcal{OP}_c\neq\emptyset$, gas-price feasibility implies
$\text{gasprice}_c=1$ since all declared gas prices equal $1$.
Thus
\[
U^{\text{sys}}
=\sum_{c\in\mathcal{CHAIN}}\text{gasprice}_c\cdot \text{Gas}_c
=\sum_{c\in\mathcal{CHAIN}}\text{Gas}_c,
\quad
\text{Gas}_c=\min(\text{Demand}_c,\text{Supply}_c).
\]
If a chain has both operators, then $\text{Supply}_c=\min(B,B)=B$, so the total processed
gas across all chains is at most $B$ and cannot reach $T=2B$; while, if a chain $c$ has exactly one operator, then $\text{Supply}_c=\text{gas}_j=B$.
Hence, any assignment with $U^{\text{sys}}\ge 2B$ must place the two operators on two
\emph{different} chains, call them $c_1$ and $c_2$, each with $\text{Supply}_{c_k}=B$.

Therefore,
\[
U^{\text{sys}}=\text{Gas}_{c_1}+\text{Gas}_{c_2}
=\min(\text{Demand}_{c_1},B)+\min(\text{Demand}_{c_2},B)
\le B+B=2B.
\]
We have $U^{\text{sys}}=2B$ iff $\text{Demand}_{c_1}\ge B$ and $\text{Demand}_{c_2}\ge B$.
Since $\text{Demand}_{c_1}+\text{Demand}_{c_2}=\sum_i a_i=2B$, this holds iff
$\text{Demand}_{c_1}=\text{Demand}_{c_2}=B$, i.e., if and only if the applications can be split into
two sets whose gas demands each sum to $B$. This is exactly a \textsc{Partition} solution.
\end{proof}

\subsection{Practical Solvability}

Despite the theoretical intractability, the optimization problem can be solved with standard tools such as mixed-integer linear programming solvers (e.g., Gurobi~\cite{gurobi}) or with metaheuristics (e.g., as implemented in Nevergrad~\cite{bennet2021nevergrad}). 
Furthermore, although the optimization is combinatorial, it has exploitable structure: feasible configurations induce sparse assignments (each agent joins at most one chain), and the search space contains large symmetries that can be collapsed by canonicalization. The problem is naturally warm-started across epochs---previous assignments provide a high-quality baseline---and supports incremental re-optimization when only a subset of bids or capacities change. In practice, engineered solver behavior (memoization, symmetry reduction, constraint-aware penalization) improves performance, and parallelization applies both to MILP solvers (e.g., commercial MIP solvers) and to population-based/metaheuristic evaluation. Since reconfiguration can be scheduled on coarse timescales (hours to a day), there is typically sufficient time for high-quality off-chain computation.

\begin{remark}
In particular, we envision two complementary implementation paths. First, a single off-chain solver can compute an allocation and attach an on-chain-verifiable certificate of feasibility and objective value (``compute off-chain, verify on-chain''); a practical caveat is that certifying auxiliary components such as ML-based forecasts or anomaly detection may be heavier than certifying the allocation itself, so such inputs may initially be treated as auditable (governance-supervised) rather than fully proven. 
Second, the protocol may support a competitive ``solution market'' in which multiple independent solver families submit candidate solutions and on-chain verification selects the best feasible one. A main subtlety in this market view is tie-handling: when multiple optimal (or near-optimal) solutions exist, selecting uniformly among them is generally intractable without enumerating the optimal set; thus uniformity is best viewed as an idealization, approximated via verifiable randomized tie-breaking over a canonical representation.
\end{remark}

\paragraph{Bilevel structure.}
Moreover, beyond generic solver improvements, the optimization admits a natural
\emph{bilevel decomposition}--the decomposition separates the problem into two levels: the outer level
assigns applications and operators to chains (discrete), while the inner
level determines the clearing gas price $p_c$ for each chain given those
assignments (continuous). Formally, the bilevel formulation is
\[
\max_{\mathbf{x},\mathbf{y}} \; F\bigl(\mathbf{x},\mathbf{y},
\mathbf{p}^*(\mathbf{x},\mathbf{y})\bigr)
\quad\text{where}\quad
\mathbf{p}^*(\mathbf{x},\mathbf{y})
\in \arg\max_{\mathbf{p}} F(\mathbf{x},\mathbf{y},\mathbf{p})
\]
subject to the assignment and price-feasibility constraints.

A key structural observation is that the objective is
\emph{chain-separable} with respect to price:
$F(\mathbf{x},\mathbf{y},\mathbf{p})
= \sum_c F_c(\mathbf{x},\mathbf{y},p_c) + G(\mathbf{x},\mathbf{y})$,
where $G$ collects all price-independent terms (including the diversity
penalty). The inner optimization therefore decomposes into independent
per-chain subproblems. Moreover, in all current model variants—core
model with or without application price penalty, overprovisioning,
capability compatibility, diversity, and stability (downtime and gas
degradation)—the per-chain objective is affine in $p_c$, so the optimum
is attained at a boundary of the feasible interval and can be computed
analytically in $O(1)$ per chain. (The only exception is the
fee-degradation extension, where $\Delta^{\text{fee}}_a$ introduces a
piecewise-linear dependence on $p_c$ with breakpoints at previous-epoch
prices; in this case the optimum is at a boundary or breakpoint, still
analytical but $O(k)$ per chain where $k$ is the number of applications. And, indeed, other extensions to the model may break this bilevel separation.)


\subsection{Multiple Optimal Solutions}
For fairness (as discussed next), when multiple optimal solutions exist, an idealized rule is to select an optimal assignment uniformly at random and define compensation relative to the resulting ex-ante expectations. Exact uniform sampling over the full optimal set is generally intractable, so in practice we approximate it via randomized tie-breaking and repeated solver runs (Monte Carlo), sampling uniformly from the distinct optimal (or $\varepsilon$-optimal) assignments encountered. This yields an auditable, expectation-based fairness proxy and enables a simple compensation rule based on the gap between realized and estimated expected utility.\footnote{Note that such compensations need not mean app-to-app monetary transfers, as these can be managed by the system using app collateral.}

\begin{remark}[Balancing fairness, efficiency, and dynamics]
While the model maximizes a weighted utility sum, different optimal assignments may distribute utility very differently across agents, with implications for fairness, security, and long-term dynamics.

\begin{itemize}
    
\item\emph{Applications.} Fewer apps per chain improve per-app satisfaction but reduce coverage; more apps per chain increase diversity but may spread resources thinly. Fairness can
be immediate (all get something each epoch) or in expectation (rotation across epochs). Both patterns risk losing applications over time, either because they are persistently under-served or because they occasionally receive nothing.

\item
\emph{Operators.} More operators per chain improve decentralization and stake security but reduce efficiency. Fairness can be immediate (yield distributed each epoch) or
in expectation (rotation over time). Expectation-based fairness risks exit after long low-yield periods.

\item
\emph{Governance levers.} The system may impose minimum service guarantees for apps, minimum stake diversity, or yield thresholds for operators to mitigate agent dropout.

\item
\emph{Future objectives.} Alternatives such as Nash product~\cite{caragiannis2019unreasonable}, egalitarian max-min, or proportional fairness could better capture these tradeoffs, and dynamic models may explicitly integrate dropout probabilities and decentralization constraints.

\end{itemize}

\end{remark}

\subsection{Properties}\label{section:properties}

The core model exhibits several key properties.

\paragraph{Adaptivity.}
Assignments are recomputed each epoch from fresh declarations. Chains thus reconfigure automatically as demand, capacity, stake, or economic preferences change.

\paragraph{Efficiency.}
The optimizer selects a configuration that maximizes the weighted objective. 
No feasible reassignment can improve the outcome under the chosen weights, 
so the solution is Pareto efficient with respect to applications, operators, and the system.

\paragraph{Fairness.}
On a single chain, applications share gas proportionally when demand exceeds supply, ensuring equal relative treatment. 

And, as mentioned, when multiple optimal assignments exist, one is chosen uniformly at random (ensuring ex-ante fairness for ops and apps alike); and, to avoid immediate disadvantage to apps, apps are compensated by the difference between their realized and expected utilities (we do not do this for operators as these are typically large infrastructure entities that can tolerate yield fluctuations across epochs).

So, we define the compensation for application $a$ as
\(
comp_a = \mathbb{E}[U^{\text{app}}_a] - U^{\text{app}}_a
\),
with positive values corresponding to compensation transfers, negative to contributions, summing to zero; and where the expectation is over the uniform distribution on all optimal assignments.\footnote{Exact expectations over all optimal assignments are generally intractable; in practice we approximate them by Monte Carlo sampling of $\varepsilon$-optimal solutions with randomized tie-breaking, using the sample mean as $\widehat{\mathbb{E}}[U^{\text{app}}_a]$. This provides an efficient, auditable estimate for compensation values.} This ensures that each application receives its utility.
Moreover,
\(
\sum_a comp_a = 0
\),
so gains and losses exactly cancel, preserving efficiency within and across epochs.

\paragraph{Incentives.}
For applications, declaring true gas demand and gas price is the safest strategy: overstating risks infeasibility or penalties backed by collateral (as the system can fine/slash applications/operators that do not deliver), understating risks exclusion or under-service. 
Operators similarly benefit from truthful declarations. 
While full strategyproofness cannot be guaranteed—agents may attempt to collude 
or manipulate bids—the clearing price depends on global assignments, 
which limits the impact of unilateral deviations. 
This aligns incentives toward honest reporting while preserving efficiency.

\begin{remark}
Note that we do not expect full (and especially group) strategyproofness under multi-dimensional declarations (demand/capacity/price/requirements) and hard feasibility constraints. We leave stronger theoretical incentive guarantees, including group strategyproofness concerns and multi-epoch equilibrium analysis to future work.

We emphasize, however, practical mitigations against adversarial behavior:
\begin{enumerate}

\item
identifiability (e.g., registry- or key-based) enables accountability and makes some forms of collusion and repeated coordinated deviations partially detectable;

\item
monitoring based on public execution traces can flag suspicious patterns (e.g., systematic slack exploitation or correlated under-reporting), with ML-based anomaly detection treated as an auxiliary, governance-supervised signal; and,

\item
misreporting gas, capabilities, or capacity can be discouraged by explicit, collateral-backed penalties (slashing) when deviations are observed.

\end{enumerate}
\end{remark}

\subsection{Example}

Consider a toy example with the following agent declarations:
\begin{align*}
&a_1:\ \gas=120,\ \stake=60,\ \gasprice^{\max}=12, \\
&a_2:\ \gas=80,\ \stake=20,\ \gasprice^{\max}=8, \\
&o_H:\ \gas=150,\ \stake=60,\ \gasprice^{\min}=6, \\
&o_L:\ \gas=150,\ \stake=20,\ \gasprice^{\min}=9.
\end{align*}

\paragraph{Assignment A (both apps with $o_H$ on chain $0$).}
Price interval: $6 \le \gasprice_0 \le 8$, then a feasible gas price is, e.g., $\gasprice=8$.
Demand$_0 =200$, Supply$_0 =150$, so proportional service:
$a_1$ gets $90$, $a_2$ gets $60$.
Fees: $8\cdot150=1200$.
Operator yield: $1200/60=20$.

\paragraph{Assignment B (only $a_1$ with $o_H$).}
Price interval: $6 \le p \le 12$, take $p=12$.
Demand $=120$, supply $=150$, so $a_1$ fully served with $120$.
Fees: $12\cdot120=1440$.
Operator yield: $1440/60=24$.

\paragraph{Discussion.}
Assignment~A serves both apps but gives lower operator yield.
Assignment~B drops $a_2$, raises the price, and improves yield.
Weights are crucial: Application-centric weights favor A (as the total application utility is greater) while operator-centric weights favor~B.
More generally, observe the decentralization level regarding application and operator utilities distribution.

\section{Modular Extensions}

The core model described captures the essential tradeoff between applications, operators, and the system. In practice, however, real blockchains introduce additional requirements: future gas demands may not be declared accurately, applications rely on specific execution features, ecosystems benefit from diversity of app types, and reconfiguration across epochs carries costs. To address these, we add modular extensions. Each can be introduced independently, preserving core properties. In Section~\ref{section:unified} we demonstrate their addition to the overall optimization.

The extensions presented below were selected as significant representatives of the modular optimization design. Further relevant extensions are possible and are discussed in Section~\ref{section:outlook}.

\subsection{Gas Overprovisioning}\label{section:overprovisioning}

Applications may not be able to predict their exact future gas demand, leading to potential under-declaration and sudden bursts of activity that cannot be served. To mitigate this, we incorporate gas overprovisioning into the allocation via a tunable slack parameter $\gamma>1$, 
used to cushion under-declaration and absorb short-term bursts. 

\paragraph{Slack-adjusted chain gas.}
For each chain $c\in\mathcal{CHAIN}$, redefine the overprovisioned gas as
\[
\text{Gas}^{(\gamma)}_c \;=\; \min\!\bigl(\,\gamma\cdot \text{Demand}_c,\; \text{Supply}_c\,\bigr)\,.
\]

Note that this affects \emph{only} the application base utility:
\[U^{\text{app}}_a \;=\;
\sum_{c\in\mathcal{CHAIN}}
x_{a,c}\;\frac{\text{gas}_a}{\text{Demand}_c}\;\text{Gas}^{(\gamma)}_c\ .\]

\paragraph{Remarks.}
\begin{itemize}
    \item \textbf{Baseline.} When $\gamma=1$, the model reverts to the baseline behavior.
    \item \textbf{Trade-off.} When $\gamma>1$, the system provisions more than the total declared demand. This can inflate per-app utility even for truthful declarations and may lead to underutilization of operator resources. The key trade-off is robustness to under-declaration versus efficiency.
    \item \textbf{Calibration.} Accurate learning or calibration of $\gamma$ (e.g., from historical fill rates or ML forecasts) is critical for balancing fairness and resource utilization.
\end{itemize}

\subsection{Capability Compatibility}

Applications may require specific computational features, while operators differ in what they support: different execution environments (e.g., EVM), privacy mechanisms (e.g., based on zero-knowledge proofs), or throughput modes (e.g., secure but slower vs.\ fast-finality). We encode these as capability vectors, and a chain is feasible only if operator capabilities satisfy all application requirements on it.  

Formally, we fix a set of capability dimensions indexed by $d \in \{1,\dots,D\}$. Each application $a \in \mathcal{APP}$ declares a capability vector $\text{cap}_{a,d} \in \{-1,0,1\}$, where: $1$ ($0$, $-1$) indicates a need for a capability enablement (disablement, indifference, resp.).
Each operator $o \in \mathcal{OP}$ declares a capability vector 
$\text{cap}_{o,d} \in \{-1,0,1\}$, where: $1$ ($0$, $-1$) indicates support (lack of support, indifference, resp.).

Each chain $c \in \mathcal{CHAIN}$ is associated with a binary capability $\text{cap}_{c,d} \in \{0,1\}$ for each dimension $d$. Feasibility requires that:
\begin{itemize}
    \item If application $a$ is assigned to chain $c$, then $$\text{cap}_{a,d} = 1 \;\Rightarrow\; \text{cap}_{c,d}=1$$
    and 
    $$\text{cap}_{a,d} = 0 \;\Rightarrow\; \text{cap}_{c,d}=0\ .$$
    \item If operator $o$ is assigned to chain $c$, then
    $$\text{cap}_{c,d}=1 \;\Rightarrow\; \text{cap}_{o,d} \in \{1,-1\}$$ and $$\text{cap}_{c,d}=0 \;\Rightarrow\; \text{cap}_{o,d} \in \{0,-1\}\ .$$
\end{itemize}

\subsection{Application-Type Diversity}

Ecosystem health depends on a balanced mix of application domains, not just concentration in one use case (e.g., perpetual futures). We model this through a target distribution over application types and penalize deviations, capturing governance goals for resilience/diversity. 

Let $\mathcal{T}=\{1,\dots,T\}$ be the set of application types, and let each application $a \in \mathcal{APP}$ have a declared type $\text{type}_a \in \mathcal{T}$.\footnote{Types can be in principle be selected via ML and governance.} Governance specifies a target diversity vector $\vec{d}^{\text{target}}=(d^{\text{target}}_1,\dots,d^{\text{target}}_T)$ with $\sum_{t=1}^T d^{\text{target}}_t=1$. Given a concrete assignment, let 
$d^{\text{real}}_t = \tfrac{1}{|\mathcal{APP}|}\sum_{a \in \mathcal{APP}} \mathbf{1}[\text{type}_a = t]$ for all $t \in \mathcal{T}$ and define the penalty as 
$\text{DiversityPenalty} = \|\vec{d}^{\text{real}} - \vec{d}^{\text{target}}\|_1$, which decreases as the realized mix of applications approaches the desired distribution.

\subsection{Epoch-to-Epoch Stability}

Reassignments between epochs enable adaptation but cause downtime, operator reconfiguration costs, and service degradation. For instance, an app may be temporarily offline during chain migration, or an operator may reinstall a new execution environment. We capture these effects with historical assignment variables and penalties, stabilizing assignments at the expense of some efficiency.  

\paragraph{Historical variables and derived changes.}
Let $x^{\text{prev}}_{a,c},y^{\text{prev}}_{o,c}\in\{0,1\}$ denote previous-epoch assignments. For each chain $c$, define $$\text{AddedApps}_c=\sum_a\max(0,x_{a,c}-x^{\text{prev}}_{a,c})$$ and $$\text{RemovedApps}_c=\sum_a\max(0,x^{\text{prev}}_{a,c}-x_{a,c})\ .$$ For each operator $o$, let $$\text{moved}_o=\mathbf{1}[\sum_c|y_{o,c}-y^{\text{prev}}_{o,c}|>0]\ .$$

\paragraph{Operator reconfiguration time and chain downtime.}
Here, we have $\delta^{\text{rem}},\delta^{\text{add}},\delta^{\text{env}}\ge 0$ as fixed parameters representing, respectively, the time to remove one application, to add one application, and to reinstall a new execution environment when an operator switches chains. 

The reconfiguration time of operator $o$ is 
\begin{align*}
    \text{Computetime}_o &= \sum_c y^{\text{prev}}_{o,c}\delta^{\text{rem}}\text{RemovedApps}_c \\ &+\sum_c y_{o,c}\delta^{\text{add}}\text{AddedApps}_c \\ &+\delta^{\text{env}}\cdot\text{moved}_o\ .
\end{align*}
The downtime of chain $c$ is $\text{Downtime}_c=\max_{o:y_{o,c}=1}\text{Computetime}_o$, propagating further as $\text{DowntimeCost}_a=\sum_c x_{a,c}\text{Downtime}_c$ and $\text{DowntimeCost}_o=\sum_c y_{o,c}\text{Downtime}_c$.

\paragraph{Service degradation penalties.}
For application $a$, let previous/current declared gas be $\text{gas}^{\text{prev}}_a,\text{gas}_a>0$ and previous and current max (declared) gas prices $\text{gasprice}^{\text{prev}}_a,\text{gasprice}_a>0$. The assigned gas is $\text{GasAssigned}_a=\sum_c x_{a,c}\,\tfrac{\text{gas}_a}{\text{Demand}_c}\text{Gas}_c$, with previous analogue $\text{GasAssigned}^{\text{prev}}_a$. The chain-clearing prices seen by $a$ are $\text{gaspriceAssigned}^{\text{prev}}_a=\sum_c x^{\text{prev}}_{a,c}\text{gasprice}^{\text{prev}}_c$ and $\text{gaspriceAssigned}_a=\sum_c x_{a,c}\text{gasprice}_c$. We are interested in the possible decline of service, thus define $$\Delta^{\text{gas}}_a=\max\{0,\tfrac{\min(\text{GasAssigned}^{\text{prev}}_a,\text{gas}^{\text{prev}}_a)}{\text{gas}^{\text{prev}}_a}-\tfrac{\min(\text{GasAssigned}_a,\text{gas}_a)}{\min(\text{gas}_a,\text{gas}^{\text{prev}}_a)}\}$$ and $$\Delta^{\text{fee}}_a=\max\{0,\tfrac{\text{gaspriceAssigned}_a}{\text{gasprice}_a}-\tfrac{\text{gaspriceAssigned}^{\text{prev}}_a}{\text{gasprice}^{\text{prev}}_a}\}\ .$$ The degradation cost is $\text{DegradationCost}_a=\alpha^{\text{gas}}\Delta^{\text{gas}}_a+\alpha^{\text{fee}}\Delta^{\text{fee}}_a$ (with $\alpha^{\text{gas}},\alpha^{\text{fee}}\ge 0$ being system-set weights determining the relative importance of gas-based and fee-based degradation).

\subsection{Unified Multiobjective Aggregation}\label{section:unified}

As the modules introduce heterogeneous quantities, we compile them into one objective. Utilities (application satisfaction, operator yield, system throughput) and penalties (downtime, degradation, diversity deviation) operate on different scales, so we normalize all terms to $[0,1]$ and invert penalties into rewards. Governance weights then control the balance across applications, operators, and the system for the complete multiobjective optimization.  
This unified formulation enables consistent comparison and aggregation of diverse metrics, ensuring that no single component dominates due to scale differences.
By adjusting these weights between epochs, governance can dynamically shift priorities—for example, emphasizing system stability during peak load or operator fairness during resource scarcity—thereby linking technical optimization to policy-level objectives.

\paragraph{Normalization.}\footnote{We use min--max scaling as a simple baseline to make heterogeneous objectives comparable. An important next step is to study more robust alternatives, including normalization via quantile- or variance-based scaling, and other techniques studied in the multiobjective optimization literature}
To enable comparability across heterogeneous terms, we normalize every quantity $Q$ to the unit interval using min–max scaling. Given conservative lower and upper bounds $Q^{\min}$ and $Q^{\max}$ (derived analytically, from stress-tested simulations, or from historical runs), we define $$\text{Norm}(Q) \;=\; \frac{Q - Q^{\min}}{\,Q^{\max} - Q^{\min}\,}$$ and $$\overline{\text{Norm}}(Q) \;=\; 1 - \text{Norm}(Q)\ .$$
Thus, utilities are expressed as positive rewards in $[0,1]$, and penalties are transformed into rewards via their complement. This ensures all components share a common range and orientation, allowing them to be combined transparently in the multiobjective optimization.

\paragraph{Applications}
Formally, we define the final utility of each agent type. For applications, first, We define the gas price cost of an app relative to the clearing gas price of the chain it is assigned to (so $c$ is the chain that the app $a$ is assigned to):
\[
\text{PriceCost}_{\text{a}} \;=\;
\frac{\text{gasprice}_{\text{c}}}{\text{gasprice}_{\text{a}}}.
\]
Note that it is always in \((0,1]\) by feasibility.

Then, each $a$ specifies weights $$\omega^{\text{base}}_a,\omega^{\text{downtime}}_a,\omega^{\text{degradation}}_a,\omega^{\text{price}}_a\in[0,1]$$ with sum $1$, and its overall utility is then
\begin{align*}
U^{\text{final}}_a \;=\;& \omega^{\text{base}}_a \,\text{Norm}(\text{AppBaseUtil}_a) \\
&+ \omega^{\text{price}}_a \,\overline{\text{Norm}}(\text{PriceCost}_a) \\
&+ \omega^{\text{downtime}}_a \,\overline{\text{Norm}}(\text{DowntimeCost}_a) \\
&+ \omega^{\text{degradation}}_a \,\overline{\text{Norm}}(\text{DegradationCost}_a)\ .
\end{align*}

\paragraph{Operators}
For operators, utility is measured directly in monetary terms: it reflects fees earned per epoch per unit of staked capital, adjusted for downtime; in particular, $\text{OpBaseUtil}_o / n_o$ represents lost monetary utility per epoch, thus multiplying it by $\text{Downtime}_o$ represents lost utility due to downtime per epoch. Thus,
\[
U^{\text{final}}_o=\tfrac{1}{\text{stake}_o}\Bigl(\text{OpBaseUtil}_o-\tfrac{\text{Downtime}_o \cdot \text{OpBaseUtil}_o}{n_o}\Bigr),
\]
where $\text{OpBaseUtil}_o$ is the operator’s fee revenue (denoted $U_o^{op}$ above), $\text{Downtime}_o$ is the downtime of its chain, and $n_o$ is the average number of epochs between configuration changes for operators (as can be estimated empirically from recent history or via ML forecasts). This adjustment penalizes operators proportionally when reconfiguration reduces availability. 

\paragraph{System}
For the system, utility aggregates throughput, stake attraction, and diversity with weights $\omega^{\text{sys}}_{\text{core}},\omega^{\text{sys}}_{\text{stake}},\omega^{\text{sys}}_{\text{diversity}}$ (summing to $1$): \begin{align*}
U^{\text{sys}} \;=\;& \omega^{\text{sys}}_{\text{core}} \,\text{Norm}(\text{SysBaseUtil}) \\
&+ \omega^{\text{sys}}_{\text{stake}} \,\text{Norm}(\text{TotalStake}) \\
&+ \omega^{\text{sys}}_{\text{diversity}} \,\overline{\text{Norm}}(\text{DiversityPenalty})\ .
\end{align*}

\paragraph{The optimization expression}
Finally, governance defines global weights $\lambda_{\text{app}},\lambda_{\text{op}},\lambda_{\text{sys}}\ge 0$ with $\lambda_{\text{app}}+\lambda_{\text{op}}+\lambda_{\text{sys}}=1$. The unified optimization problem is $\max\;\;\lambda_{\text{app}}\cdot \tfrac{1}{|\mathcal{APP}|}\sum_a U^{\text{final}}_a+\lambda_{\text{op}}\cdot \tfrac{1}{|\mathcal{OP}|}\sum_o \text{Norm}(U^{\text{final}}_o)+\lambda_{\text{sys}}\cdot U^{\text{sys}}$.  

\paragraph{Interpretation.}
This formulation compiles all modules into a single governance-weighted optimization problem. It balances application satisfaction, operator incentives, and system-level goals, while preserving modularity and ensuring consistent integration.

\section{Simulations}

We describe the simulation performed and report on the results. The goal of the simulations is to demonstrate its ability to control the utilities of the different agent types.
Note that we use synthetic data both because real-world datasets are not yet available and because controlled inputs let us demonstrate the governance effects clearly.

Next we discuss the specific model considered, the specific solver used in the simulations, and the instances that were solved.

\subsection{Specific Model Used}

Next we describe the specific model used, which is the core model with these changes:
\begin{itemize}

\item
\textbf{Mid-range gas price.}
We set the clearing price of a chain \(price_c\) to be the mean between the minimum declared price over the applications and the maximum required price over the operators; i.e.: $$price_c = \Bigl(\min_{a\in \mathcal{APP}_c} price_a \;+\; \max_{o\in \mathcal{OP}_c} price_o\Bigr) / 2.$$

\item
\textbf{Price-penalized application utility.}
To capture the idea that applications prefer lower gas prices (as in the general model), we introduce a mild penalty term that depends on the ratio between the actual clearing price \(price_c\) and the application's declared maximum price \(price_a\).
The ratio \(\frac{price_c}{price_a}\in[0,1]\) ensures that the penalty is \(0\) when the chain price is \(0\) and maximal when the chain price equals the app's declared bound.

For normalization, to keep the utility bounded within \([0,1]\), we define\footnote{We use $0.1$ to cap the utility cost of gasprice at $\%10$. While this is rather arbitrary it can be selected by governance; in the context of this paper it does not affect the general behavior.}
\[
U_a = 0.1 + 0.9 \cdot utilization \;-\; 0.1 \cdot \frac{price_c}{price_a},
\]
where \(utilization\) denotes the share of the app's declared gas that was actually processed.
As gas is used proportionally between the apps, this coincides with the chain utilization:
\[
utilization = \frac{Gas_c}{\text{Demand}_c}
             = \frac{\text{total gas computed on chain}}{\text{total gas demand on chain}}.
\]
The constants \(0.1\) and \(0.9\) provide a small baseline and ensure that \(U_a\in[0,1]\).

\item
\textbf{Operator and system normalization.}
We use $Q^{\min}=0$ throughout. For the \emph{system} (total fees),
a direct feasibility upper bound is
\[
Q^{\text{sys}}_{\max}
\;=\;
\Bigl(\min\bigl(\sum_{o} \gas_o,\ \sum_{a} \gas_a\bigr)\Bigr)\cdot \max_{a} p_{\max}(a),
\]
i.e., total processed gas is bounded by aggregate supply/demand, and the clearing price on any
chain is bounded by the highest declared application cap.

For \emph{operators}, utility is measured as yield (fees per unit stake):
in the core model, any operator assigned to chain $c$ obtains
$U^{\text{op}}_o=\text{Fee}_c/\text{Stake}_c$.
Thus, an appropriate global upper bound for normalization is the maximum feasible total fees
divided by the minimum feasible chain stake:
\[
Q^{\text{op}}_{\max}
\;=\;
\frac{Q^{\text{sys}}_{\max}}{\min_{o}\text{stake}_o},
\]
since for every nonempty chain $c$ we have $\text{Stake}_c=\sum_{o\in\mathcal{OP}_c}\text{stake}_o
\ge \min_{o}\text{stake}_o$ and hence $\text{Fee}_c/\text{Stake}_c\le Q^{\text{sys}}_{\max}/\min_o \text{stake}_o$.
In the simulations we normalize system utility by $Q^{\text{sys}}_{\max}$ and operator yield by
$Q^{\text{op}}_{\max}$.

\end{itemize}

\subsection{Specific Solver Used}

We implemented the optimization model in Python using the \texttt{Nevergrad} library for derivative-free multiobjective optimization. Our simulation framework allows generating synthetic instances of applications, operators, and chains, specifying their gas, stake, and price parameters, and evaluating utilities and constraints under different governance weights. For now, we focus on uniformly random instances to test scalability and qualitative behavior (while indeed stylized, such instances capture the essential combinatorial structure of the allocation problem; richer empirically grounded generators are planned for future work); also, for simplicity, we report on simulations done on the core model (without the extensions). We carefully engineer the solver, as described next:
\begin{itemize}

\item
\textbf{Memoization.} We cache previously evaluated configurations and skip recomputation whenever the optimizer revisits the same point in the search space. This eliminates redundant evaluations and significantly accelerates convergence, especially when the optimizer samples locally or revisits identical combinations. In later stages, we combine this with canonicalization (below) to also detect equivalent configurations under relabeling.

\item
\textbf{Canonicalization.} Since chain identifiers are arbitrary, many allocations are equivalent up to relabeling. To avoid exploring symmetric duplicates, each assignment is converted into a canonical form by renumbering chains contiguously according to their first appearance—akin to a restricted-growth sequence representation. This ensures that equivalent configurations share the same memoization key, effectively collapsing the search space.\footnote{The implementation of this canonicalization resembles a \emph{restricted-growth sequence} (RGS); this general concept provides a canonical encoding of a set partition by assigning integer labels in order of first appearance, thereby removing dependence on arbitrary group identifiers~\cite{rgs}.}

\item
\textbf{Nevergrad optimization.} 
We use the standard Nevergrad optimizer (ParametrizedOnePlusOne).

\end{itemize}

\subsection{Specific Instances Solved}

\paragraph{Instance.}
We consider a single application \(a\) and three operators
\(o_H,o_L^{(1)},o_L^{(2)}\),
with the following parameters:

\begin{itemize}
\item
Application \(a\): demand \(D=100\), price cap \(p_{\max}=10\).
\item
Operators:
each operator has capacity \(50\) (so total capacity \(=150>100\)).
Two \emph{low-floor} operators have \(p_{\min}=0\),
and one \emph{high-floor} operator has \(p_{\min}=\kappa\,p_{\max}\)
for a parameter \(\kappa\in(0,1]\)
(the \emph{price-spread knob}).
\item
Stakes (sum to \(1\)):
the high-floor operator is a \emph{whale} (or the opposite direction) with stake \(\sigma\in(0,1)\)
(the \emph{stake-skew knob}),
and each low-floor operator has stake \((1-\sigma)/2\).

\end{itemize}

We use $\kappa = 0.3$ and $\sigma = 0.9$. Thus, we normalize with \(Q^{sys}_{\max}=D\cdot p_{\max}=1000\); and \(Q^{op}_{\max}=Q^{sys}_{\max} / (1 - \sigma) / 2) = 1000 / ((1 - 0.9) / 2) = 20000\).

\subsection{Simulation Results}

\begin{figure*}[t]
  \centering
  \includegraphics[width=\linewidth]{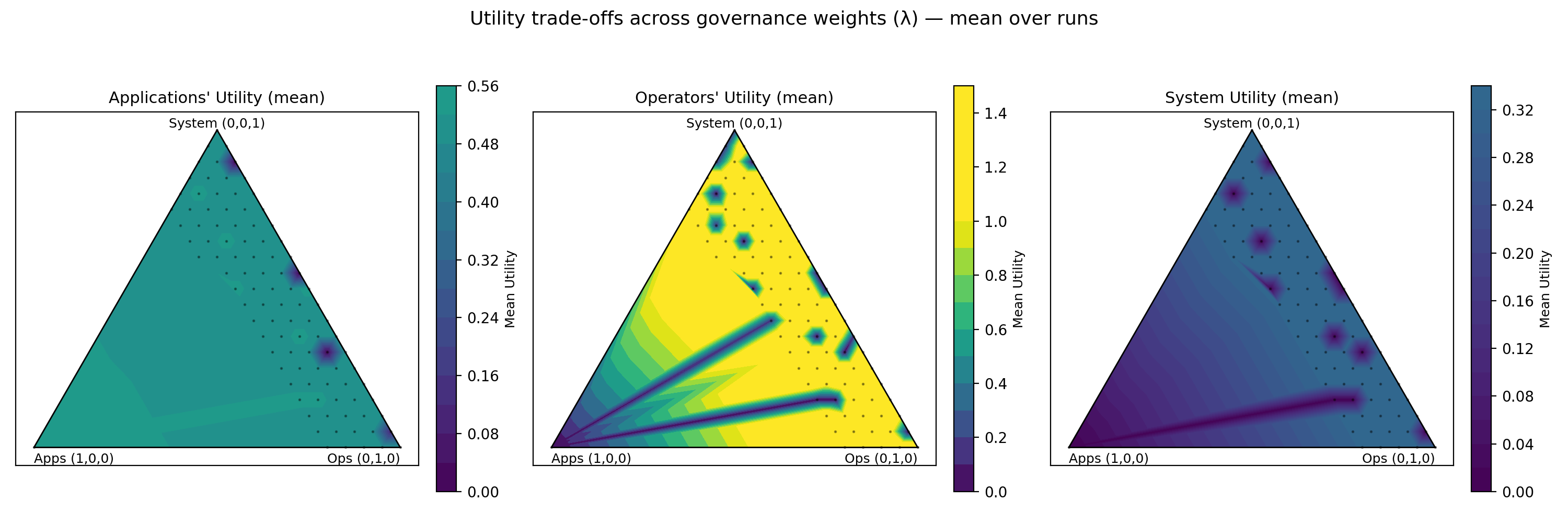}
  \caption{\textbf{Governance control over utilities.} Each triangular panel is a ternary plot over the governance-controlled parameters of multiagent preferences $(\lambda_{\text{app}},\lambda_{\text{op}},\lambda_{\text{sys}})$. From left to right: (i) application utility, (ii) operator utility, (iii) system utility. Color indicates normalized utility (colder~$\to$~lower, warmer~$\to$~higher). Utilities peak near the vertex where the corresponding weight dominates
  (e.g., app utility near $(1,0,0)$), and trade off smoothly along the edges and interior.}
  \label{fig:governancecontrol}
\end{figure*}

To illustrate how governance weights steer outcomes, we sweep the simplex of
$(\lambda_{\text{app}},\lambda_{\text{op}},\lambda_{\text{sys}})$ and record the realized utilities
for applications, operators, and the system. Because the weights satisfy
$\lambda_{\text{app}}+\lambda_{\text{op}}+\lambda_{\text{sys}}=1$, the search space is a 2D simplex,
visualized as an equilateral triangle: the bottom-left, bottom-right, and top vertices correspond to
$(1,0,0)$, $(0,1,0)$, and $(0,0,1)$, respectively. Each point in the simplex corresponds to a different weights, controlling the multiobjective optimization. 

Figure~\ref{fig:governancecontrol} shows the results.

\section{Misalignment in Allocation}\label{section:misalignment}

The simulations above illustrate that governance weights may steer outcomes substantially:
shifting $(\lambda_{\text{app}},\lambda_{\text{op}},\lambda_{\text{sys}})$ redistributes
utility across stakeholder types. A natural follow-up question is:
\emph{when does this steering matter more/less?}
When stakeholders have aligned preferences, governance is trivial---any
reasonable weight vector yields near-identical outcomes, and fine-tuning
weights is a low-stakes exercise.
Conversely, when preferences conflict fundamentally, governance faces
unavoidable tradeoffs: improving one stakeholder's outcome necessarily
degrades another's, and weight selection has real welfare consequences.
\emph{Misalignment}, discussed and formally defined next, quantifies this structural tension for a given
instance, independent of governance's weight choice.

\subsection{Normative Desiderata}

We seek a measure satisfying:
\begin{enumerate}
    \item \textbf{Instance-specific:} It characterizes a given instance
    $I$, not the governance choice of $\lambda$.
    \item \textbf{Scale-invariant:} It uses relative utilities---the
    fraction of solo-best achieved, where ``solo-best'' denotes the
    utility a stakeholder type would obtain if the optimization were
    run solely on its behalf.
    \item \textbf{Interpretable extremes:}
    Zero misalignment means some assignment satisfies all stakeholders
    optimally; high misalignment means no assignment gives all
    stakeholders substantial fractions of their solo-best.
\end{enumerate}
This cleanly separates \textbf{governance policy} (choosing $\lambda$,
a normative and ultimately political decision that may be informed by analysis as done above) from \textbf{instance
hardness} (misalignment, a structural property of the declared
preferences and feasibility constraints).

\subsection{Formal Definitions}

\paragraph{Setup.}
Fix an instance $I$ with feasible assignments $\mathcal{X}(I)$.\footnote{Misalignment is a property of a particular instance. Given a set of instances or a distribution of instances, one may sample from the distribution to get an estimation of the expected misalignment of instances from that distribution.}
Each stakeholder type $k \in \{\text{app}, \text{op}, \text{sys}\}$ has
aggregate utility $U^k(x)$ for $x \in \mathcal{X}(I)$ (for applications
and operators, this is the average utility over agents of that type).

\paragraph{Solo optimum.}
The best achievable utility for type $k$ in instance $I$:
\[
U^k_*(I) \;:=\; \max_{x \in \mathcal{X}(I)}\; U^k(x).
\]

\paragraph{Relative utility.}
The fraction of the solo-best that type $k$ achieves under assignment $x$:
\[
r^k(x) \;:=\; \frac{U^k(x)}{U^k_*(I)} \;\in\; [0,1].
\]

\paragraph{Global misalignment.}
\[
\text{mis}(I) \;:=\; 1 \;-\; \max_{x \in \mathcal{X}(I)}\;
\min_{k \in \{\text{app},\text{op},\text{sys}\}}\; r^k(x)
\;\in\; [0,1].
\]
Thus $\text{mis}(I) = 0$ if and only if some assignment simultaneously
maximizes all three utilities; $\text{mis}(I) \to 1$ when every feasible
assignment leaves at least one stakeholder type near-zero relative
utility.

\subsection{Properties and Interpretation}

\paragraph{Directionality.}
Low misalignment is desirable: it means the instance is ``easy to
govern''---almost any reasonable weight vector yields an outcome where no
stakeholder is badly underserved.
High misalignment signals a structurally contentious instance:
governance must navigate genuine tradeoffs, and in the extreme
($\text{mis}(I) \approx 1$), any assignment necessarily leaves some
stakeholder type near-zero relative utility---a zero-sum-like regime.

\paragraph{Governance attention signal.}
Misalignment provides an actionable diagnostic:
where misalignment is low, governance is forgiving---even naive or
politically-driven weight choices do little harm;
where misalignment is high, governance is critical and fine-tuning
weights has real welfare consequences.
Epochs or instances with high misalignment thus deserve more deliberate
weight calibration, while low-misalignment epochs can be governed
on autopilot.

\paragraph{Misalignment-minimizing governance.}
One can define a meta-objective where, rather than fixing weights, the system selects the assignment that directly minimizes misalignment---i.e.,
maximizes the worst-case relative utility (so, instead of getting the assignment with best utility given some weights, the goal here is to the find the assignment that minimizes the efective misalignment):
\[
x^* \;\in\; \arg\max_{x \in \mathcal{X}(I)}\;
\min_{k}\; r^k(x).
\]
This is an egalitarian max-min rule over relative utilities: a
governance policy that is explicitly instance-adaptive.
It is one natural policy among many; the point is that misalignment
provides the vocabulary to reason about and compare such choices.

\begin{remark}[Instance geometry]
Misalignment is fundamentally about the shape of the Pareto frontier in
relative-utility space. When the frontier passes close to the point
$(1,1,1)$, misalignment is low; when it bows away, misalignment is high.
The governance weight vector $\lambda$ selects a point on this frontier
but cannot reshape it---which is precisely why governance is
``external'' and political in nature: it chooses among the options the
instance affords, but cannot create options that do not exist.
\end{remark}

\begin{remark}[Misalignment dynamics]
Misalignment is not a static property of the system---it evolves as
agents enter or exit, change their declarations, or as market conditions
shift.
An epoch where a high-stake operator demands elevated prices while
applications are price-sensitive will exhibit higher misalignment than an
epoch where preferences happen to be compatible.

Tracking misalignment over time provides a structural diagnostic of the
system's health:
persistently high misalignment may signal fundamental market imbalance
(e.g., chronic under-supply of capacity or systematically mismatched
capabilities),
while transient spikes may reflect temporary shocks that the adaptive
mechanism can absorb in subsequent epochs.
This connects directly to the dynamic analysis direction outlined in
Section~\ref{section:outlook}: understanding how agent adaptation---learning,
entry, exit, and strategic adjustment---shapes misalignment trajectories
across epochs is a natural and important next step.
\end{remark}

\section{Outlook}\label{section:outlook}

We have formulated the problem of adaptively configuring a multichain blockchain infrastructure as a modular multiagent optimization problem and we have demonstrated initial feasibility and tradeoff control. Future research directions include:
\begin{itemize}
  \item \textbf{Empirical observation.} Once implemented in practice (in progress), the framework can be evaluated on real deployments—observing actual declarations of gas, stake, and pricing, and validating that the optimization behaves as expected under real-world demand and operator behavior.
  \item \textbf{Better optimization.} Develop faster and more adaptive optimization methods, including incremental updates between epochs and warm-starting from previous solutions.
  \item \textbf{Extensive simulation.} Evaluate scalability through systematic benchmarking—varying the number of applications, operators, and capability dimensions—and reporting computational runtimes across hardware configurations and solver budgets. Complement this with richer instance generators that reflect realistic blockchain workloads, including heavy-tailed demand and capacity distributions, correlated price–stake profiles, moving beyond the synthetic instances used here.
  \item \textbf{Richer resources.} Extend operators to multiple roles (e.g., validators, provers, archival nodes).
  \item \textbf{Interconnectivity.} Extend the model to capture composability and interdependence among applications, allowing them to express, e.g., co-location preferences.
  \item \textbf{Pricing and MEV.} Support non-uniform fee-to-gas within a chain while mitigating strategic ordering (e.g., order-fairness/MEV controls) and studying welfare–fairness impacts.
  \item \textbf{Stability and latency.} Add explicit stability/latency objectives and ablations on their weights.
  \item \textbf{Incentive analysis.} Formally study the strategyproofness of the mechanism—both for individual operators and coalitions—characterizing the scope for misreporting (gas capacity, costs, stake) and colluding on pricing or availability.
  \item \textbf{Dynamic analysis.} Study the system’s evolution across epochs using both simulation and game-theoretic models, examining how adaptive reassignments, compensation, and agent learning shape long-term equilibria and participation dynamics towards a robust multiagent framework.
\end{itemize}

\bibliography{bib}
\bibliographystyle{plain}

\end{document}